\pdfoutput=1

\documentclass[11pt, letterpaper]{article}

\usepackage{fancyhdr}

\usepackage{amssymb}
\usepackage{graphicx}
\usepackage{geometry}

\usepackage{url}
\urldef{\mailsa}\path|{amaecker, malatya, fmadh, soerenri}@hni.upb.de|

\usepackage[utf8]{inputenc}
\usepackage[american]{babel}
\usepackage{xspace}
\usepackage{enumerate}
\usepackage{xargs}
\usepackage{mathtools} 
\usepackage{amsmath}
\usepackage{amsthm}
\usepackage{nicefrac}
\usepackage{algorithm}
\usepackage{hyperref}

\usepackage[capitalise,nameinlink,noabbrev]{cleveref}
\crefname{ineq}{Inequality}{Inequalities}
\creflabelformat{ineq}{#2{\upshape(#1)}#3}
\crefname{enumi}{}{}
\crefname{step}{Step}{Steps}
\creflabelformat{step}{(#2#1#3)}
\crefname{lemma}{Lemma}{Lemmas}
\crefname{corollary}{Corollary}{Corollaries}
\crefname{proposition}{Proposition}{Propositions}

\newcommand*{\ALG}{\ensuremath{\textsc{Balance}}\xspace}
\newcommand*{\OPT}{\ensuremath{\textsc{Opt}}\xspace}
\newcommand*{\oho}{\ensuremath{\text{overhead}_\OPT}\xspace}
\newcommand*{\oha}{\ensuremath{\text{overhead}_\ALG}\xspace}

\newcommand*{\comb}{\ensuremath{\text{com}}\xspace}
\newtheorem{theorem}{Theorem}[section]
\newtheorem{lemma}[theorem]{Lemma}
\newtheorem{corollary}[theorem]{Corollary}
\newtheorem{proposition}[theorem]{Proposition}

\begin{document}

\title{Non-Clairvoyant Scheduling to Minimize Max Flow Time on a Machine with Setup Times\thanks{This work was partially supported by the German Research Foundation (DFG)
within the Collaborative Research Centre ``On-The-Fly Computing'' (SFB 901)}~~\footnote{A conference version of this paper has been accepted for publication in the proceedings of the ``15th Workshop on Approximation and Online Algorithms'' (WAOA).}}

\author{Alexander M\"acker \and Manuel Malatyali \and Friedhelm Meyer auf der Heide \and S\"oren Riechers \\ [0.4em]
	Heinz Nixdorf Institute \& 	Computer Science Department\\
	Paderborn University, Germany\\[0.2em]
	\{amaecker, malatya, fmadh, soerenri\}@hni.upb.de}
\date{}

\maketitle

\begin{abstract}
Consider a problem in which $n$ jobs that are classified into $k$ types arrive over time at their release times and are to be scheduled on a single machine so as to minimize the maximum flow time.
The machine requires a setup taking $s$ time units whenever it switches from processing jobs of one type to jobs of a different type.
We consider the problem as an online problem where each job is only known to the scheduler as soon as it arrives and where the processing time of a job only becomes known upon its completion (non-clairvoyance).

We are interested in the potential of simple ``greedy-like'' algorithms.
We analyze a modification of the FIFO strategy and show its competitiveness to be $\Theta(\sqrt{n})$, which is optimal for the considered class of algorithms.
For $k=2$ types it achieves a constant competitiveness.
Our main insight is obtained by an analysis of the smoothed competitiveness.
If processing times $p_j$ are independently perturbed to $\hat p_j = (1+X_j)p_j$, 
we obtain a competitiveness of $O(\sigma^{-2} \log^2 n)$ when $X_j$ is drawn from a uniform or a (truncated) normal distribution with standard deviation $\sigma$.
The result proves that bad instances are fragile and ``practically'' one might expect a much better performance than given by the $\Omega(\sqrt{n})$-bound.
\end{abstract}

\section{Introduction}
Consider a scheduling problem in which there is a single machine for processing jobs arriving over time.
Each job is defined by a release time at which it arrives, a size describing the time required to process it and it belongs to exactly one of $k$ types.
Whenever the machine switches from processing jobs of one type to jobs of a different type, a setup needs to take place for the reconfiguration of the machine.
During a setup the machine cannot process any workload.
A natural objective in such a model is the minimization of the time each job remains in the system.
This objective was introduced in \cite{flowtime} as maximum flow time, defined as the maximum time a job spends in the system, that is, the time between the arrival of a job and its completion.
It describes the quality of service as, for example, perceived by users and aims at schedules being responsive to each job.
In settings in which user interaction is present or processing times may depend on other further inputs not known upon the arrival of a job, it is also natural to assume the concept of non-clairvoyance, as introduced in \cite{nonclairvoyant}.

There are several applications for a model with job types and setup times mentioned in the literature.
Examples are settings in which a server has to answer requests of different types, depending on the data to be loaded into memory and to be accessed by the server \cite{saksReport,operator}; manufacturing systems, in which machines need to be reconfigured or cleaned during the manufacturing of different customer orders \cite{manufacturing}; or a setting where an intersection of two streets is equipped with traffic lights and where setup times describe the time drivers need for start-up once they see green light \cite{traffic}.

In this paper, we study the potential of ``greedy-like'' online algorithms in terms of their (smoothed) competitiveness.
The formal model and notions are given in \cref{sec:model}.
In \cref{sec:algo} we analyze the competitiveness of ``greedy-like'' algorithms and show matching upper and lower bounds of $\Theta(\sqrt{n})$, where the bound is achieved by a simple modification of the \emph{First In First Out} (FIFO) strategy.
For the special case of $k=2$ types, the competitiveness improves to $O(1)$.
Our main result is an analysis of the smoothed competitiveness of this algorithm in \cref{sec:smoothed}, which is shown to be $O(\sigma^{-2} \log^2 n)$ where $\sigma$ denotes the standard deviation of the underlying smoothing distribution.
It shows worst case instances to be fragile against random noise and that, except on some pathological instances, the algorithm achieves a much better performance than suggested by the worst case bound on the competitiveness.

\section{Model \& Notions}
\label{sec:model}
We consider a scheduling problem in which $n$ jobs, partitioned into $k$ types, are to be scheduled on a single machine. 
Each job $j$ has a size (processing time) $p_j \in \mathbb{R}_{\geq 1}$, a release time $r_j \in \mathbb{R}_{\geq 0}$ and a parameter $\tau_j$ defining the type it belongs to.
The machine can process at most one job at a time.
Whenever it switches from processing jobs of one type to a different one and before the first job is processed, a setup taking constant $s$ time units needs to take place during which the machine cannot be used for processing.
The goal is to compute a non-preemptive schedule, in which each job runs to completion without interruption once it is started, that minimizes the maximum flow time $F \coloneqq \max_{1 \leq j \leq n} F_j$ where 
$F_j$ is the amount of time job $j$ spends in the system.
That is, a job $j$ arriving at $r_j$, started in a schedule at $t_j$ and completing its processing at $c_j \coloneqq t_j + p_j$ has a flow time $F_j \coloneqq c_j-r_j$.

Given a schedule, a \emph{batch} is a sequence of jobs, all of a common type $\tau$, that are processed in a contiguous interval without any intermediate setup.
For a batch $B$, we use $\tau(B)$ to denote the common type $\tau$ of $B$'s jobs and $w(B) \coloneqq \sum_{j \in B} p_j$ to denote its workload.
We refer to setup times and idle times as \emph{overhead} and overhead is associated to a job $j$ if it directly precedes $j$ in the schedule.
For an interval $I = [a,b]$ we also use $l(I) \coloneqq a$ and $r(I) \coloneqq b$ and $w(I) \coloneqq \sum_{j : r_j \in I} p_j$ to denote the workload released in interval $I$.

\paragraph*{Non-Clairvoyant Greedy-like Online Algorithms}
We consider our problem in an \emph{online} setting where jobs arrive over time at their release times and are not known to the scheduler in advance.
Upon arrival the scheduler gets to know a job together with its type but does not learn about its processing time, which is only known upon its completion (\emph{non-clairvoyance}) \cite{nonclairvoyant}.
We are interested in the potential of conceptually simple and efficient \emph{greedy-like} algorithms.
For classical combinatorial offline problems, the concept of greedy-like algorithms has been formalized by Borodin et al.\ in \cite{priority} by \emph{priority algorithms}.
We adopt this concept and for our online problem we define greedy-like algorithms to work as follows:
When a job completes (and when the first job arrives), the algorithm determines a total ordering of \emph{all possible} jobs without looking at the actual instance.
It then chooses (among already arrived yet unscheduled jobs) the next job to be scheduled by looking at the instance and selecting the job coming first according to this ordering.

\paragraph*{Quality Measure}
To analyze the quality of online algorithms, we facilitate competitive analysis.
It compares solutions of the online algorithm to solutions of an optimal offline algorithm which knows the complete instance in advance.
Precisely, an algorithm \textsc{Alg} is called $c$-\emph{competitive} if, on any instance $\mathcal{I}$, $F(\mathcal{I}) \leq c \cdot F^*(\mathcal{I})$, where $F(\mathcal{I})$ and $F^*(\mathcal{I})$ denote the flow time of \textsc{Alg} and an optimal (clairvoyant) offline solution on instance $\mathcal{I}$, respectively.

Although competitive analysis is the standard measure for analyzing online algorithms, it is often criticized to be overly pessimistic.
That is, a single or a few pathological and very rarely occurring instances can significantly degrade the quality with respect to this measure.
To overcome this, various alternative measures have been proposed in the past (e.g.\ see \cite{compAlter1,compAlter2,compAlter3}).
One approach introduced in \cite{smoothedComp} is \emph{smoothed competitiveness}.
Here the idea is to slightly perturb instances dictated by an adversary by some random noise and then analyze the expected competitiveness, where expectation is taken with respect to the random perturbation. 
Formally, if input instance $\mathcal{I}$ is smoothed according to some smoothing (probability) distribution $f$ and if we use $N(\mathcal{I})$ to denote the instances that can be obtained by smoothing $\mathcal{I}$ according to $f$, the smoothed competitiveness $c_\text{smooth}$ is defined as $c_\text{smooth} \coloneqq \sup_\mathcal{I} \mathbb{E}_{\mathcal{\hat I} \overset{f}{\gets} N(\mathcal{I})} \left[\frac{F(\mathcal{\hat I)}}{F^*(\mathcal{\hat I})}\right]$.
We will smoothen instances by randomly perturbing processing times.
We assume the adversary to be \emph{oblivious} with respect to perturbations.
That is, the adversary constructs the instance based on the knowledge of the algorithm and $f$ (so that $\mathcal{I}$ is defined at the beginning and is not a random variable).

\section{Related Work}
The problem supposedly closest related to ours is presented in a paper by Divakaran and Saks \cite{saks}.
They consider the clairvoyant variant in which the processing time of each job is known upon arrival.
Additionally, they allow the setup time to be dependent on the type.
For this problem, they provide an $O(1)$-competitive online algorithm. 
Also, they show that the offline problem is NP-hard in case the number $k$ of types is part of the input.
In case $k$ is a constant, it was known before that the problem can be solved optimally in polynomial time by a dynamic program proposed by Monma and Potts in \cite{monma}.
When all release times are identical, then the offline problem reduces to the classical makespan minimization problem with setup times.
It has been considered for $m$ parallel machines and it is known \cite{wads} to be solvable by an FPTAS if $m$ is assumed to be constant; 
here an (F)PTAS is an approximation algorithm that finds a schedule with makespan at most by a factor of $(1+\varepsilon)$ larger than the optimum in time polynomial in the input size (and $\frac{1}{\varepsilon}$).
For variable $m$, Jansen and Land propose a PTAS for this problem in \cite{land}. 
Since in general a large body of literature for scheduling with setup considerations has evolved over time, primarily in the area of operations research, the interested reader is referred to the surveys by Allahverdi et al.\ \cite{ali1,ali2,ali3}.

Our model can also be seen as a generalization of classical models without setup times.
In this case, it is known that FIFO is optimal for minimizing maximum flow time on a single machine.
On $m$ parallel machines FIFO achieves a competitiveness of $3-2/m$ (in the (non-)preemptive case) as shown by Mastrolilli in \cite{fifo}.
Further results include algorithms for (un-)related machines with speed augmentation given by Anand et al.\ in \cite{journalMaxFlowtime} and for related machines proposed by Bansal and Cloostermans in \cite{journalMaxFlowtime2}.

The concept of smoothed analysis has so far, although considered as an interesting alternative to classical competitiveness (e.g.\ \cite{compAlter1,compAlter2,compAlter3}), only been applied to two problems.
In \cite{smoothedComp}, Bechetti et al.\ study the Multilevel Feedback Algorithm for minimizing total flow time on parallel machines when preemption is allowed and non-clairvoyance is assumed. 
They consider a smoothing model in which initial processing times are integers from the interval $[1, 2^K]$ and are perturbed by replacing the $k$ least significant bits by a random number from $[1, 2^k]$.
They prove a smoothed competitiveness of $O((2^k/\sigma)^3+(2^k/\sigma)^2 2^{K-k})$, where $\sigma$ denotes the standard deviation of the underlying distribution.
This, for example, becomes $O(2^{K-k})$ for the uniform distribution.
This result significantly improves upon the lower bounds of $\Omega(2^K)$ and $\Omega(n^{\frac{1}{3}})$ known for the classical competitiveness of deterministic algorithms \cite{nonclairvoyant}.
In \cite{taskSystems}, Sch\"afer and Sivadasan apply smoothed competitive analysis to metrical task systems (a general framework for online problems covering, for example, the paging and the $k$-server problem).
While any deterministic online algorithm is (on any graph with $n$ nodes) $\Omega(n)$-competitive, the authors, amongst others, prove a sublinear smoothed competitiveness on graphs fulfilling certain structural properties.
Finally, a notion similar to smoothed competitiveness has been applied by Scharbrodt, Schickinger and Steger in \cite{averageComp}.
They consider the problem of minimizing the total completion time on parallel machines and analyze the Shortest Expected Processing Time First strategy.
While it is $\Omega(n)$-competitive, they prove an expected competitiveness, defined as $\mathbb{E}\left[\frac{\textsc{Alg}}{\OPT}\right]$, of $O(1)$ if processing times are drawn from a gamma distribution.

\section{A Non-Clairvoyant Online Algorithm}
\label{sec:algo}
In this section, we present a simple greedy-like algorithm and analyze its competitiveness.
The idea of the algorithm \ALG, as presented in \cref{fig:alg}, is to find a tradeoff between preferring jobs with early release times and jobs that are of the type the machine is currently configured for.
This is achieved by the following idea: Whenever the machine is about to idle at some time $t$, \ALG checks whether there is a job $j$ available that is of the same type $\tau_j$ as the machine is currently configured for, denoted by $active(t)$.
If this is the case and if there is no job $j'$ with a ``much smaller'' release time than $j$, job $j$ is assigned to the machine. 
The decision whether a release time is ``much smaller'' is taken based on a parameter $\lambda$, called \emph{balance parameter}.
This balance parameter is grown over time based on the maximum flow time encountered so far and, at any time, is of the form $\alpha^q$, for some $q \in \mathbb{N}$ which is increased over time and some constant $\alpha$ determined later.\footnote{A variant of this algorithm with a fixed $\lambda$ and hence without \cref{step:grow} has been previously mentioned in \cite{saks} as an algorithm with $\Omega(n)$-competitiveness for the clairvoyant variant of our problem.}
Note that \ALG is a greedy-like algorithm by using the adjusted release times for determining the ordering of jobs.

\begin{algorithm}
\begin{enumerate}[(1)]
\item Let $\lambda = \alpha$. \hfill $\vartriangleright \text{for some constant } \alpha$
\item If the machine idles at time $t$,\\ process available job with smallest \textbf{adjusted release time} $\bar r_j(t)$
\[ \bar r_j(t) \coloneqq
  \begin{cases}
    r_j       & \quad \text{if } \tau_j = \text{active}(t)\\
    r_j + \lambda  & \quad \text{else}\\
  \end{cases}
\]
after doing a setup if necessary.\\
To break a tie, prefer job $j$ with $\tau_j = \text{active}(t)$.
\item As soon as a job $j$ completes with $F_j \geq \alpha \lambda$, set $\lambda \coloneqq \alpha \lambda$. \label[step]{step:grow}
\end{enumerate}
\caption{Description of \ALG}
\label{fig:alg}
\end{algorithm}
\subsection{Basic Properties of \ALG}
The following two properties follow from the definition of \ALG and relate the release times and flow times of consecutive jobs, respectively.
For a job $j$, let $\lambda(j)$ denote the value of $\lambda$ when $j$ was scheduled.
\begin{proposition}
\label{prop:releaseDiff}
Consider two jobs $j_1$ and $j_2$.
If $\tau_{j_1} = \tau_{j_2}$, both jobs are processed according to FIFO.
Otherwise, if $j_2$ is processed after $j_1$ in a schedule of \ALG, $r_{j_2} \geq r_{j_1} - \lambda(j_1)$.
\end{proposition}
\begin{proof}
The first statement directly follows from the definition of the algorithm.
Consider the statement for two jobs $j_1$ and $j_2$ with $\tau_{j_1} \neq \tau_{j_2}$.
Let $t$ be the point in time at which $j_1$ is assigned to the machine.
If $\text{active}(t) \neq \tau_{j_1}$ and $\text{active}(t) \neq \tau_{j_2}$, it follows $r_{j_2} \geq r_{j_1}$.
If $\text{active}(t) = \tau_{j_1}$, then because $j_1$ is preferred over $j_2$, we have $r_{j_1} = \bar r_{j_1}(t) \leq \bar r_{j_2}(t) = r_{j_2}+\lambda(j_1)$.
Finally, if $\text{active}(t) = \tau_{j_2}$ we know by the fact that $j_1$ is preferred that $r_{j_1} + \lambda(j_1) = \bar r_{j_1}(t) < \bar r_{j_2}(t) = r_{j_2}$, which proves the proposition.
 \end{proof}

\begin{proposition}
\label{prop:flowtime}
Consider two jobs $j_1$ and $j_2$.
If $j_1$ is processed before $j_2$ and no job is processed in between, then $F_{j_2} \leq F_{j_1} + p_{j_2} + s + \lambda(j_1)$.
\end{proposition}

\begin{proof}
First note that \ALG does not idle deliberately.
Hence, if there is idle time between the processing of job $j_1$ and $j_2$, then $t_{j_2} \leq r_{j_2} + s$ holds.
Thus, we have $F_{j_2} \leq s + p_{j_2}$ proving the claim.

If there is no idle time, by definition $j_1$ is finished by time $r_{j_1}+F_{j_1}$.
Since $j_2$ is processed directly afterward, it is finished not later than $r_{j_1} + F_{j_1} +s + p_{j_2}$.
By Proposition~\ref{prop:releaseDiff} this is upper bounded by $r_{j_2}+ \lambda(j_1) + F_{j_1} +s + p_{j_2}$, which proves the desired bound.
\end{proof}

\subsection{Competitiveness}
We carefully define specific subschedules of a given schedule $S$ of \ALG, which we will heavily use throughout our analysis of the (smoothed) competitiveness.
Given $\alpha^q \geq F^*$, $q \in \mathbb{N}_0$, let $S_{\alpha^q}$ be the subschedule of $S$ that starts with the first job $j$ with $\lambda(j) = \alpha^q$ 
and ends with the last job $j'$ with $\lambda(j') = \alpha^q$.
For a fixed $\delta$, let $S^\delta_{\alpha^q}$ be the suffix of $S_{\alpha^q}$ such that the first job in $S^\delta_{\alpha^q}$ is the last one in $S_{\alpha^q}$ with the following properties:
(1) It has a flow time of at most $(\alpha-\delta)\alpha^q$, and
(2) it starts a batch.
(We will prove in \cref{le:workload} that $S^\delta_{\alpha^q}$ always exists.)
Without loss of generality, let $j_1, \ldots, j_m$ be the jobs in $S^\delta_{\alpha^q}$ such that they are sorted by their starting times, $t_1 < t_2 < \ldots < t_m$.
Let $B_1, \ldots, B_\ell$ be the batches in $S^\delta_{\alpha^q}$.
The main idea of \cref{le:workload} is to show that, in case a flow time of $F > \alpha^{q+1}$ is reached, the interval $[r_{j_1}, r_{j_m}]$ is in a sense dense:
Workload plus setup times in $S_{\alpha^q}^\delta$ is at least by $\delta \alpha^q$ larger than the length of this interval.
Intuitively, this holds as otherwise the difference in the flow times $F_{j_1}$ and $F_{j_m}$ could not be as high as $\delta \alpha^q$, which, however, needs to hold by the definition of $S_{\alpha^q}^\delta$.
Additionally, the flow time of all jobs is shown to be lower bounded by $3\alpha^q$.
Roughly speaking, this holds due to the following observation:
If a fixed job has a flow time below $3\alpha^q$, then the job starting the next batch can, on the one hand, not have a much smaller release time (by definition of the algorithm).
On the other hand, it will therefore not be started much later, leading to the fact that the flow time cannot be too much larger than $3\alpha^q$ (and in particular, is below $(\alpha-\delta)\alpha^q$ for sufficiently small $\delta$).

\begin{lemma}
\label{le:workload}
Let $\alpha^q \geq F^*$ and $\delta \leq \alpha - 10$.
Then $S_{\alpha^q}^\delta$ always exists and all jobs in $S_{\alpha^q}^\delta$ have a flow time of at least $3\alpha^q$.
Also, if $F > \alpha^{q+1}$, it holds
$
  \sum_{i=1}^{\ell}w(B_i) + r_{j_1} - r_{j_m} \geq \delta \alpha^q - (\ell-1)s
$.
\end{lemma}

\begin{proof}
We first prove that a job with the two properties starting off $S_{\alpha^q}^\delta$ exists.
Let $\tilde{j}_1,\ldots$ be the jobs in $S_{\alpha^q}$.
Consider the last job $\tilde{j}_0$ processed directly before $\tilde{j}_1$.
By Proposition~\ref{prop:flowtime} we have $F_{\tilde{j}_1} \leq F_{\tilde{j}_0} + p_{\tilde{j}_1} + s + \alpha^{q-1} \leq \alpha^q + p_{\tilde{j}_0} + s + \alpha^{q-1} + p_{\tilde{j}_1} + s + \alpha^{q-1} < 6 \alpha^q$.
Among jobs in $S_{\alpha^q}$ that have a different type than $\tilde{j}_1$, consider the job $\tilde{j}_i$ with the lowest starting time.
We show that it is a candidate for starting $S_{\alpha^q}^\delta$, implying that $S_{\alpha^q}^\delta$ exists.
Property (2) directly follows by construction.
For the flow time of $\tilde{j}_i$, we know that only jobs of the same type as $\tilde{j}_1$ are scheduled between $\tilde{j}_1$ and $\tilde{j}_i$.
This implies that jobs $\tilde{j}_2,\ldots,\tilde{j}_{i-1}$ are released in the interval $[r_{\tilde{j}_1},r_{\tilde{j}_i}+\alpha^q]$.
The interval can contain a workload of at most $r_{\tilde{j}_i}+\alpha^q - r_{\tilde{j}_1} + F^*$ (see also \cref{obs:opt}), hence the flow time of job $\tilde{j}_i$ is at most $F_{j_1} = F_{\tilde{j}_i} \leq (r_{\tilde{j}_1}+F_{\tilde{j}_1}+(r_{\tilde{j}_i}+\alpha^q - r_{\tilde{j}_1} + F^*) + s + p_{j_i}) - r_{\tilde{j}_i} \leq 9 \alpha^q + p_{\tilde{j}_i} = 9 \alpha^q + p_{j_1}\leq (\alpha-\delta)\alpha^q$.
Property (1) and the existence of $S_{\alpha^q}^\delta$ follow.
Since $t_{j_1} = c_{j_1} - p_{j_1}$ and $F_{j_1} = c_{j_1} - r_{j_1}$, we also have $t_{j_1} \leq r_{j_1} + 9 \alpha^q \leq r_{j_1} + (\alpha-\delta)\alpha^q$ (*).

We now show that during $S_{\alpha^q}^\delta$, the machine does not idle and each job in $S_{\alpha^q}^\delta$ has a flow time of at least $3\alpha^q$.
Assume this is not the case.
Denote by $t$ the last time in $S_{\alpha^q}^\delta$ where either an idle period ends or a job with a flow time of less than $3\alpha^q$ completes.
We denote the jobs scheduled after $t$ by $\hat{j}_1,\ldots$ and the first job of the first batch started at or after $t$ by $\hat{j}_i$.
Similar to above, all jobs $\hat{j}_1,\ldots,\hat{j}_i$ are released in the interval $[t-3\alpha^q,r_{\hat{j}_i}+\alpha^q]$.
The overall workload of these jobs is at most $r_{\hat{j}_i}+4\alpha^q - t + F^* \leq r_{\hat{j}_i}+5\alpha^q - t$.
Job $\hat{j}_i$ is thus finished by $t + (r_{\hat{j}_i}+5\alpha^q - t) + s \leq r_{\hat{j}_i}+6\alpha^q$.
This is a contradiction to $F_{\hat{j}_i} > (\alpha-\delta)\alpha^q$.

Finally, since there are no idle times and by (*), for the last job $j_m$ of $S^\delta_{\alpha^q}$ we have $F_{j_m} \leq r_{j_1} + (\alpha-\delta)\alpha^q + \sum_{i=1}^{\ell}w(B_i) + (\ell-1)s - r_{j_m}$.
By the assumption that $F > \alpha^{q+1}$ and the definition of $j_m$ to be the first job with flow time at least $\alpha^{q+1}$, we obtain the desired result.
\end{proof}
We will also make use of \cref{cor:workload}, which follows from the proof of \cref{le:workload}.

\begin{corollary}
\label{cor:workload}
The statement of \cref{le:workload} also holds if $S_{\alpha^q}^\delta$ is replaced by $S_{\alpha^q}^\delta(j)$ for any job $j \in S_{\alpha^q}^\delta$ with $F_j \leq (\alpha-\delta)\alpha^q$, where $S_{\alpha^q}^\delta(j)$ is the suffix of $S_{\alpha^q}^\delta$ starting with job $j$.
\end{corollary}

Next we give simple lower bounds for the optimal flow time $F^*$.
Besides the direct lower bound $F^* \geq \max\{s, p_{max}\}$, where $p_{max} \coloneqq \max_{1\leq j \leq n} p_j$, we can also prove a bound as given in \cref{obs:opt}.
For a given interval $I$, let $\oho(I)$ be the overhead in \OPT between the jobs $j_1$ and $j_2$ released in $I$ and being processed first and last in \OPT, respectively.
Precisely, $j_1 \coloneqq \text{argmin}_{j : r_j \in I} t_j$ and $j_2 \coloneqq \text{argmax}_{j : r_j \in I} t_j$.

\begin{proposition}
\label{obs:opt}
As lower bounds for $F^*$ we have $F^* \geq \max\{s, p_{max}\}$ as well as 
$
F^* \geq \max_I\{w(I) + \oho(I) - |I|\}$.
\end{proposition}

\begin{proof}
We have $F^* \geq c_{j_2} -r_{j_2}$.
On the other hand, $c_{j_2} \geq r_{j_1} + \oho(I) + w(I)$.
Thus, $F^* \geq \oho(I) + w(I) + l(I) - r(I) = w(I) + \oho(I) - |I|$.
\end{proof}

Combining \cref{le:workload} and \cref{obs:opt}, we easily obtain that the competitiveness can essentially be bounded by the difference in the number of setups \OPT and \ALG perform on those jobs which are part of $S^\delta_{\alpha^q}$.
Let $I(S_{\alpha^q}^\delta)$ be the interval in which all jobs belonging to $S_{\alpha^q}^\delta$ are released, $I(S_{\alpha^q}^\delta) \coloneqq [\min_{j}\{r_j: j \in S_{\alpha^q}^\delta\}, \max_{j}\{r_j : j  \in S_{\alpha^q}^\delta\}]$.
We have the following bound.

\begin{lemma}
\label{le:competitiveness}
Let $\alpha^{q+1} \leq F < \alpha^{q+2}$ and $3\leq \delta \leq \alpha-10$ and $\alpha^q \geq F^*$.
It holds
$
    F \leq \alpha^2(\delta-2)^{-1}(F^* + \oha(S^\delta_{\alpha^q}) -\oho(I(S^\delta_{\alpha^q})))
$.
\end{lemma}

\begin{proof}
Suppose to the contrary that it holds
$\alpha^{q} > (\delta-2)^{-1}(F^* + \oha(S^\delta_{\alpha^q}) -\oho(I(S^\delta_{\alpha^q})))$.
By \cref{prop:releaseDiff} we have $I(S_{\alpha^q}^\delta) \subseteq [r_{j_1}-\alpha^q, r_{j_m} + \alpha^q]$ and using \cref{obs:opt} we obtain a contradiction as 
\begin{align*}
F^* & \geq w(I(S_{\alpha^q}^\delta)) + \oho(I(S_{\alpha^q}^\delta)) + r_{j_1} - r_{j_m} - 2\alpha^q \\
& \geq \sum_{i=1}^{\ell}w(B_i) + \oho(I(S^\delta_{\alpha^q})) + r_{j_1} - r_{j_m} - 2\alpha^q\\
& \overset{(\cref{le:workload})}{\geq} \delta \alpha^q - (\ell-1)s + \oho(I(S^\delta_{\alpha^q})) - 2\alpha^q > F^*,
\end{align*}
where the last inequality follows from our assumption.
 \end{proof}

Throughout the rest of the paper, we assume that $\delta = 3$ and $\alpha =13$ fulfilling the properties of \cref{{le:competitiveness},{le:workload}}.
Our goal now is to bound the competitiveness by upper bounding the difference of the overhead of \OPT and \ALG in $S_{\alpha^q}^\delta$ for some $\alpha^q = \Omega(\sqrt{n\cdot s \cdot p_{max}})$.
In \cref{le:necessaryCond} we will see that to obtain a difference of $i\cdot s$, a workload of $\Omega(i\cdot \alpha^q)$ is required.
Using this, we can then upper bound the competitiveness based on the overall workload of $O(n \cdot p_{max})$ available in a given instance in \cref{th:comp}.
Before we can prove \cref{le:necessaryCond} we need the following insight.
Given $S^\delta_{\alpha^q}$ for some $q \in \mathbb{N}_0$ such that $\alpha^q \geq F^*$.
Let $j_{\tau,i}$ be the first job of the $i$-th batch of some fixed type $\tau$ in $S^\delta_{\alpha^q}$.
We show that the release times of jobs $j_{\tau,i}$ and $j_{\tau, i+1}$ differ by at least $\alpha^q$.
Intuitively, this holds due to the definition of the balance parameter and the fact that in $S^\delta_{\alpha^q}$ all jobs starting a batch have a flow time of at least $3\alpha^q$.

\begin{lemma}
\label{le:releaseDiff}
Given $S^\delta_{\alpha^q}$, it holds $r_{j_{\tau,i}} > r_{j_{\tau, i-1}} + \alpha^q$, for all $i \geq 2$ and all $\tau$.
\end{lemma}
\begin{proof}
Consider a fixed job $j_{\tau,i}$ and suppose to the contrary that $r_{j_{\tau, i}} \leq r_{j_{\tau, i-1}} + \alpha^q$ holds.
As each job in $S^\delta_{\alpha^q}$ that is the first of a batch (except the very first such job) has a flow time of at least $3\alpha^q$, job $j_{\tau, i-1}$ is not started before $r_{j_{\tau, i}}$ (otherwise it would be finished not later than $r_{j_{\tau, i-1}} + \alpha^q + p_{j_{\tau, i}} \leq r_{j_{\tau,i-1}} + \alpha^q + F^* \leq r_{j_{\tau, i-1}} + 2\alpha^q$ with flow time smaller $3\alpha^q$).
Also, because $j_{\tau, i-1}$ is the first job of a batch, all jobs $j$ processed later fulfill $r_j \geq r_{j_{\tau, i-1}}$.
But then at the time $t$ at which the $(i-1)$-th batch is finished, $\bar r_{j_{\tau, i}}(t) \leq  r_{j_{\tau, i-1}} + \alpha^q \leq r_j + \alpha^q = \bar r_j(t)$ and hence, $j_{\tau,i}$ would be preferred over all such jobs $j$ and thus would belong to the same batch as $j_{\tau, i-1}$.
This contradicts the definition of $j_{\tau,i}$, proving the lemma.
 \end{proof}
\begin{lemma}
\label{le:necessaryCond}
Let $B$ be a batch in \OPT.
If all jobs from $B$ are part of $S_{\alpha^q}^\delta$, an overhead of at most $2s$ is associated to them in the schedule of \ALG.

Also, if the overhead associated to $B$ in \OPT is smaller than $2s$ and is $2s$ in the schedule of \ALG, it needs to hold 
\begin{enumerate}
    \item $w(B) \geq \alpha^q- F^* -s \eqqcolon \bar w$ and
    \item jobs of $B$ with size at least $\bar w$ need to be released in an interval of length $\alpha^q$.
\end{enumerate}
\end{lemma}
\begin{proof}
Assume to the contrary that \ALG processes the jobs of $B$ in three batches with $j_1, j_2, j_3 \in B$ being the jobs starting the first, the second and the third batch, respectively.
Then there need to be two jobs $i_1$ and $i_2$ that are processed between the first and second and second and third such batch, respectively.
Since $j_2$ is preferred over $i_2$ and by \cref{le:releaseDiff}, we have $r_{i_2} \geq r_{j_2} \geq r_{j_1} + \alpha^q$.
Also, since $i_2$ is preferred over $j_3$ and by \cref{le:workload}, we have $r_{i_2} + \alpha^q \leq r_{j_3}$.
Hence, \OPT cannot process $i_2$ before nor after $B$ (since then either $j_1$ or $i_2$ would have a flow time larger than $F^*$), which is a contradiction to the fact that $B$ is a batch in \OPT.

If \ALG processes the jobs of $B$ in two batches, let $j_1, j_2 \in B$ be the jobs starting the first batch and the second batch, respectively.
We start with the case that $w(B) < \bar w$ and show a contradiction.
We know that $r_{j_2} > r_{j_1} + \alpha^q$.
Consider an optimal schedule.
As $j_1$ cannot be started after $r_{j_1}+F^*$ and because \OPT processes $j_1$ and $j_2$ in the same batch $B$, the processing of $B$ needs to cover the interval $[r_{j_1} + F^*, r_{j_2}] \supseteq [r_{j_1} + F^*, r_{j_1}+ \alpha^q]$.
As $w(B) < \alpha^q - F^* -s$ this implies an additional overhead of at least $s$ associated to $B$, contradicting our assumption.

Therefore, assume that $w(B) \geq \bar w$ but there is no interval of length $\alpha^q$ with jobs of $B$ of size at least $\bar w$.
We know that $j_1$ needs to be started not later than $r_{j_1} + F^*$.
Also, the workload of jobs of $B$ released until $r_{j_1} + \alpha^q$ is below $\bar w$.
Hence, there needs to be a job in $B$ released not before $r_{j_1} + \alpha^q$.
This implies that the processing of $B$ needs to cover the entire interval $[r_{j_1} + F^*, r_{j_1} + \alpha^q]$.
However, this implies an additional overhead associated to $B$ of at least $s$, contradicting our assumption.
 \end{proof}

We are now ready to bound the competitiveness of \ALG.

\begin{theorem}
\label{th:comp}
\ALG is $O(\sqrt n)$-competitive.
Additionally, it holds $F = O(F^* + \sqrt{n p_{max} s})$.
\end{theorem}
\begin{proof}
If $F \leq \sqrt{n \cdot s \cdot p_{max}}$ holds, we are done as $F^* \geq \sqrt{s\cdot p_{max}}$ by \cref{obs:opt}.

Hence, consider the case where $F > \sqrt{n \cdot s \cdot p_{max}}$ and assume $\alpha^{q+1} \leq F < \alpha^{q+2}$.
Also we can assume $F^* \leq \frac{F}{\alpha^3} < \alpha^{q-1}$ as otherwise we obtain a constant competitiveness.
Consider $S^\delta_{\alpha^q}$.
We call a batch $B$ of \OPT \emph{short} if $w(B) < \bar w$ and \emph{long} otherwise.
According to \cref{le:necessaryCond}, we know that the overhead associated to jobs belonging to short batches is not larger in a schedule of \ALG than in \OPT.
On the other hand, overhead associated to jobs belonging to long batches can be at most by $s$ larger in \ALG than in \OPT.
However, as a long batch requires a workload of $\bar w = \alpha^q-F^*-s \geq \alpha^q-2F^* \geq \alpha^q-2\alpha^{q-1} \geq \frac{\sqrt{n\cdot s\cdot p_{max}}}{2\alpha^2}$, there can be at most $O(\sqrt{n} \cdot \sqrt{\frac{p_{max}}{s}})$ many long batches as $n$ jobs can have a workload of at most $n \cdot p_{max}$.
Hence, by \cref{le:competitiveness} we obtain the desired result.
 \end{proof}

For the case $k=2$ we can even strengthen the statement of \cref{le:releaseDiff}.
Given $S^\delta_{\alpha^q}$, let job $j_i$ be the first job of the $i$-th batch in $S^\delta_{\alpha^q}$ and note that $\tau_{j_i} = \tau_{j_{i+2}}$ as the batches form an alternating sequence of the two types.
We have the following lemma.
\begin{lemma}
\label{le:releaseDiff2}
Given $S^\delta_{\alpha^q}$, if $k=2$ then it holds $r_{j_i} > r_{j_{i-1}} + \alpha^q$, for all $i \geq 3$.
\end{lemma}
\begin{proof}
Consider a fixed job $j_i$ with $i \geq 3$ and suppose to the contrary that $r_{j_i} \leq r_{j_{i-1}} + \alpha^q$ holds.
By definition of $S^\delta_{\alpha^q}$, job $j_{i-1}$ is not started before $r_{j_i}$ and all jobs processed later have a release time not smaller than $r_{j_{i-1}}$.
Hence, by the definition of \ALG, $j_i$ would belong to the same batch as $j_{i-2}$, which is a contradiction.
 \end{proof}
Based on this fact, we can show that \OPT can essentially not process any jobs that belong to different batches in $S^\delta_{\alpha^q}$ in one batch.
Hence, \OPT performs roughly the same amount of setups as \ALG does and we have the following theorem by \cref{le:competitiveness}.
\begin{theorem}
\label{th:2comp}
If $k=2$, then \ALG is $O(1)$-competitive. 
\end{theorem}
\begin{proof}
Assume $F > 2F^*$ as otherwise we are done.
Consider $S^\delta_{\alpha^q}$ such that $\alpha^q < F \leq \alpha^{q+1}$.
By \cref{le:releaseDiff2}, we know that $r_{j_{i+1}} > r_{j_i}+2F^*$ for all $i \in [2,\ell]$.
Now suppose to the contrary that the optimal solution processes two jobs $j_i$ and $j_{i+2}$ in the same batch.
As $j_i$ cannot be completed later than $r_{j_i}+F^*$ and job $j_{i+2}$ is not released before $r_{j_{i+2}}$, this batch needs to cover the interval $[r_{j_i}+F^*, r_{j_{i+2}}]$.
However, job $j_{i+1}$ needs to be started during the interval $[r_{j_{i+1}}, r_{j_{i+1}}+F^*] \subseteq [r_{j_i}+F^*, r_{j_{i+2}}]$, which is a contradiction.

Hence, the optimal solution cannot process any two jobs $j_i$ and $j_{i+2}$, for all $i\geq 2$, in the same batch.
By \cref{le:competitiveness} we obtain a competitiveness of $O(1)$.
 \end{proof}

To conclude this section, we show that the bound of $O(\sqrt{n})$ from \cref{th:comp} for the competitiveness of \ALG is tight and that a lower bound of $\Omega(\sqrt{n})$ holds for any greedy-like algorithm as defined in \cref{sec:model}.
This also implies that the $\Omega(\sqrt{n})$ bound holds for \ALG independent of how $\lambda$ is chosen or increased (and even if done at random).
The construction in the proof of \cref{th:lowerBound} is a generalization of a worst-case instance given in \cite{saks}.

\begin{theorem}
\label{th:lowerBound}
Any greedy-like algorithm is $\Omega(\sqrt n)$-competitive.
\end{theorem}
\begin{proof}
The adversary will be defined such that the optimum flow time is $O(1)$ while a fixed greedy-like algorithm $A$ has a flow time of $\Omega(\sqrt{n})$.
We define the adversary by specifying the instance in phases.
Let the setup time be $s=1$.
\begin{itemize}
\item During the $i$-th phase, $\sqrt{n}$ unit size jobs of two types $\tau_{i_1} \neq \tau_{i_2}$ that did not occur in any previous phase are released in $\sqrt{n}$ consecutive (discrete) time steps.
\item The first job of phase $i$ is released at time $(i-1)(\sqrt{n} + 2)$.
(Hence, $\sqrt{n}$ jobs are released in $\sqrt{n}$ time steps, then two time steps no job is released.
Afterward this pattern is repeated.)
\item The first job released in phase $i$ is of type $\tau_{i_1}$ and the second one of type $\tau_{i_2}$.
If $A$ prefers the job of type $\tau_{i_1}$, let all remaining jobs of phase $i$ be of type $\tau_{i_1}$.
If $A$ prefers the job of type $\tau_{i_2}$, let all remaining jobs of phase $i$ be of type $\tau_{i_2}$.
\end{itemize}
We analyze the flow time of \OPT and algorithm $A$.
For each phase, the optimal solution can first process the job belonging to the type of which only one job is released and afterward all remaining jobs released during the phase.
Hence, \OPT can always start the setup for phase $i$ before or at time $(i-1)(\sqrt{n} + 3)$ and has processed all jobs of phase $i-1$ at that point in time, because it executes $\sqrt{n}$ unit size jobs and needs two setups per phase.
This gives a maximum flow time of at most $5$ (tight if $A$ prefers the job of type $\tau_{i_1}$, hence \OPT prioritizes the job of type $\tau_{i_2}$ and the first job of type $\tau_{i_1}$ remains in the system for one time step until the job of type $\tau_{i_2}$ is released, for two additional time steps while the job of type $\tau_{i_2}$ is executed, and finally for two more time steps where the job itself is executed).

For $A$ we first make two observations.
(1) We can assume that for some phase $i$, $A$ neither processes the job released first nor the job released second after the job released last as otherwise $A = \Omega(\sqrt{n})$ holds.
(2) We can assume that $A$ processes all jobs of phase $i$ before any job of phase $i+1$ because no two jobs of different phases can be processed in the same batch and hence, processing a job of phase $i$ later than a job of phase $i+1$ cannot be advantageous.
By these two observations, the algorithm $A$ has to do three setups for each phase by the construction of the instance.
Thus, it finishes the last job of phase $i$ not before $i(\sqrt{n} + 3)$.
As the adversary can construct $\sqrt{n}$ phases, the last job is finished at $\sqrt{n}(\sqrt{n}+3)$ and it is released not later than $(\sqrt{n}-1)(\sqrt{n} + 2)+\sqrt{n} = \sqrt{n}(\sqrt{n} + 2) - 2$.
Hence, the flow time of $A$ is at least $\sqrt{n}(\sqrt{n}+3) - (\sqrt{n}(\sqrt{n} + 2) - 2) = \sqrt{n} + 2 = \Omega(\sqrt{n})$.
 \end{proof}

\section{Smoothed Competitive Analysis}
\label{sec:smoothed}
In this section, we analyze the smoothed competitiveness of \ALG.
We consider the following \emph{multiplicative smoothing model} from \cite{smoothedComp}.
Let $p_j$ be the processing time of a job $j$ as specified by the adversary in instance $\mathcal{I}$.
Then the perturbed instance $\mathcal{\hat I}$ is defined as $\mathcal{I}$ but with processing times $\hat p_j$ defined by $\hat p_j = (1 + X_j)p_j$ where $X_j$ is chosen at random according to the smoothing distribution.
For $0 < \varepsilon < 1$  being a fixed parameter describing the strength of perturbation, we consider two smoothing distributions.
In case of a \emph{uniform smoothing distribution}, $X_j$ is chosen uniformly at random from the interval $[-\varepsilon, \varepsilon]$.
More formally, $X_j \sim \mathcal{U}(-\varepsilon, \varepsilon)$ where $\mathcal{U}(a,b)$ denotes the continuous uniform distribution with probability density function
$f(x) = \frac{1}{b-a}$ for $a \leq x \leq b$ and $f(x)= 0$ otherwise.
Hence, for $\hat p_j$ we have $ \hat p_j \in [(1-\varepsilon)p_j, (1+\varepsilon)p_j]$.
In case of a \emph{normal smoothing distribution}, $X_j$ is chosen from a normal distribution with expectation $0$, standard deviation $\sigma=\frac{\varepsilon}{\sqrt{2.64}}$ and truncated at $-1$ and $1$.
More formally, $X_j \sim \mathcal{N}_{(-1,1)}(0, \sigma^2)$ where $\mathcal{N}_{(a,b)}(\mu, \sigma^2)$ denotes the truncated normal distribution with probability density function $f(x) = \frac{\phi(\frac{x-\mu}{\sigma})}{\sigma(\Phi(\frac{b-\mu}{\sigma})-\Phi(\frac{a-\mu}{\sigma}))}$ for $a < x < b$ and $f(x)=0$ otherwise.
Here $\phi(\cdot)$ denotes the density function of the standard normal distribution and $\Phi(\cdot)$ the respective (cumulative) distribution function.

Our goal is to prove a smoothed competitiveness of $O(\varepsilon^{-2} s \log^2 n)$.
We analyze the competitiveness by conditioning it on the flow time of \OPT and its relation to the flow time of \ALG.
Let $\mathcal{E}_\OPT^q$ be the event that $F^* \in [\alpha^q, \alpha^{q+1})$ and $\mathcal{E}_\ALG^q$ be the event that $F > c_1 \alpha^{q+1} \varepsilon^{-2} s \log^2 n$ (for a constant value of $c_1$ determined by the analysis).
Also, denote by $\mathcal{\bar E}_x^q$ the respective complementary events. 
Then for a fixed instance $\mathcal{I}$ we obtain 
\begin{align*}
\mathbb{E}_{\mathcal{\hat I} \gets N(\mathcal{I})}\left[\frac{F(\mathcal{\hat I})}{F^*(\mathcal{\hat I})}\right]  
= \sum_{q \in \mathbb{N}} \mathbb{E}&\left[\frac{F(\mathcal{\hat I})}{F^*(\mathcal{\hat I})} \mid \mathcal{E}_\OPT^q \wedge \mathcal{\bar E}_\ALG^q\right]
\cdot \Pr[\mathcal{E}_\OPT^q \wedge \mathcal{\bar E}_\ALG^q] \\ 
+ \sum_{ q= \lfloor \log_\alpha s\rfloor}^{\lceil \log_\alpha n \rceil}  \mathbb{E}&\left[\frac{F(\mathcal{\hat I})}{F^*(\mathcal{\hat I})} \mid \mathcal{E}_\OPT^q \wedge \mathcal{E}_\ALG^q \right]
\cdot \Pr[\mathcal{E}_\OPT^q \wedge \mathcal{E}_\ALG^q] \\
+ \sum_{q > \lceil \log_\alpha n \rceil}  \mathbb{E}&\left[\frac{F(\mathcal{\hat I})}{F^*(\mathcal{\hat I})} \mid \mathcal{E}_\OPT^q \wedge \mathcal{E}_\ALG^q \right]
\cdot \Pr[\mathcal{E}_\OPT^q \wedge \mathcal{E}_\ALG^q] \enspace .
\end{align*}
Note that the first sum is by definition directly bounded by $O(\varepsilon^{-2} s \log^2 n)$ and the third one by $O(\sqrt{s})$ according to \cref{th:comp}. 
Thus, we only have to analyze the second sum.
We show that we can complement the upper bound on the ratio, which can be as high as $O(\sqrt{n})$ by \cref{th:comp}, by $\Pr[\mathcal{E}_\OPT^q \wedge \mathcal{E}_\ALG^q] \leq 1/n$.
From now on we consider an arbitrary but fixed $q \geq \lfloor \log_\alpha s \rfloor$, and in the following we analyze $\Pr[\mathcal{E}_\OPT^q \wedge \mathcal{E}_\ALG^q]$.
Let $\Gamma \coloneqq \alpha^{i-1}$ such that $i$ is the largest integer with $\alpha^i \leq c_1\varepsilon^{-2}\alpha^{q+1} s \log^2 n$.
Thus we have $\Gamma \geq c_1 \alpha^{q-1} \varepsilon^{-2} s \log^2 n$.

On a high level, the idea of our proof is as follows:
We first define a careful partitioning of the time horizon into consecutive intervals (\cref{sec:partitioning}).
Depending on the amount of workload released in each such interval and an estimation of the amount of setups required for the respective jobs (\cref{sec:setupEstimatation}), we then classify each of them to either be dense or sparse (\cref{sec:events}).
We distinguish two cases depending on the number of dense intervals in $\mathcal{I}$.
If this number is sufficiently large, $F^*$ is, with high probability (w.h.p.), not much smaller than $F$ (\cref{le:highOpt}).
This holds as w.h.p.\ the perturbation increases the workload in a dense interval so that even these jobs cannot be scheduled with a low flow time by \OPT.
In case the number of dense intervals is small, the analysis is more involved.
Intuitively, we can show that w.h.p.\ there is only a logarithmic number of intervals between any two consecutive sparse intervals in which the perturbation decreases the workload to a quite small amount.
Between such sparse intervals the flow time cannot increase too much (even in the worst-case) and during a sparse interval \ALG can catch up with the optimum:
If taking a look at the flow time of the job completing at time $t$ and continuing this consideration over time, we then obtain a sawtooth pattern always staying below a not too large bound for the flow time of \ALG (\cref{le:longRun}).

\subsection{Partitioning of Instance \texorpdfstring{$\mathcal{I}$}{I}}
\label{sec:partitioning}
We define a partitioning of the instance $\mathcal{I}$, on which our analysis of the smoothed competitiveness will be based on.
We partition the time interval $[r_{min}, r_{max}]$, where $r_{min}$ and $r_{max}$ are the smallest and largest release time, as follows:
Let a \emph{candidate interval} $C$ be an interval such that $|C| = \Gamma$ and such that for some $\tau$ it holds $\sum_{j : r_j \in C, \tau_j = \tau} p_j \geq \Gamma/4$.
Intuitively, a candidate interval $C$ is an interval on which, in $\mathcal{\hat I}$, \ALG possibly has to perform more setups than \OPT does (which, if all jobs released in the interval belong to $S^\delta_\Gamma$ and under the assumption that $\mathcal{E}_\OPT^q \wedge \mathcal{E}_\ALG^q$ holds, according to \cref{le:necessaryCond} requires a workload of at least $\frac{\Gamma}{2}$ in $\mathcal{\hat I}$ and hence, at least $\frac{\Gamma}{4}$ in $\mathcal{I}$).
Let $C_1$ be the first candidate interval $C$.
For $i>1$ let $C_i$ be the first candidate interval $C$ that does not overlap with $C_{i-1}$.

Now we consider groups of $\mu \coloneqq \left\lceil \frac{\varepsilon^2 \Gamma}{c_2 s^2\log^2n}\right\rceil$ many consecutive candidate intervals $C_i$, for some constant $c_2$ determined by the further analysis.
Precisely, these groups are defined as $I_1 = [r_{min}, r(C_\mu)]$, $I_2 = (r(C_\mu), r(C_{2\mu})]$ and so on.
In the rest of the paper we consistently use $I_i$ to denote these intervals.
Let $\bigcup_i I_i = [r_{min},r_{max}]$ by (possibly) extending the last $I_i$ so that its right endpoint is $r_{max}$.
Although it worsens constants involved in the competitiveness, we use $\mu \leq \frac{2\varepsilon^2 \Gamma}{c_2 s^2\log^2n}$ for $c_1 \geq \alpha c_2$ for the sake of simplicity.

\subsection{Estimation of Setups in $I_i$}
\label{sec:setupEstimatation}
\begin{algorithm}[t]
Construct a sequence  $(j_1, j_2, \ldots, j_m)$ of all jobs released in $I$ as follows:
\begin{enumerate}[(1)]
\item For $i=1,2,\ldots, m$ set $j_i$ to be job $j \notin (j_1, \ldots j_{i-1})$ with smallest $\bar r_j$, where 
\[ \bar r_j \coloneqq
  \begin{cases}
    r_j       & \quad \text{if } \tau_j = \tau_{j_{i-1}}\\
    r_j + \Gamma  & \quad \text{else.}\\
  \end{cases}
\]
To break a tie, prefer job $j$ with $\tau_j = \tau_{j_{i-1}}$.
\item Let $N_s(I)$ be the number of values $i$ such that $\tau_{j_i} \neq \tau_{j_{i-1}}$.
\end{enumerate}
\caption{Description of \textsc{SetupEstimate($I$)}}
\label{fig:alg2}
\end{algorithm}
Before we can now classify intervals $I_i$ to be dense or sparse, we need an estimate $N_s(I)$ on the number of setups \OPT and \ALG perform on jobs released in a given interval $I$. 
We require $N_s(I)$ to be a value uniquely determined by the instance $\mathcal{I}$ and hence, in particular not to be a random variable.
This is essential for our analysis and avoids any computation of conditional probabilities.
For the definition of $N_s(I)$ consider the construction by \textsc{SetupEstimate($I)$} in \cref{fig:alg2}.
For a fixed interval $I$, it essentially mimics \ALG in $S^\delta_{\Gamma}$ in the sense that \cref{le:necessaryCond2} holds completely analogous to \cref{le:necessaryCond}.
Also, note that the construction is indeed invariant to job sizes and hence to perturbations.
It should not be understood as an actual algorithm for computing a schedule, however, for ease of presentation we refer to the sequence constructed as if it was a schedule.
Particularly, we say that it processes two jobs $j_i$ and $j_{i'}$ with $\tau_{j_i} = \tau_{j_{i'}}$ in different batches if there is an $i''$ such that $i<i''<i'$ with $\tau_{j_{i''}} \neq \tau_{j_i}$.

For two jobs $j_1$ and $j_2$ of a common type $\tau$ which start two batches in \textsc{SetupEstimate}($I$), $r_{j_2} \geq r_{j_1} + \Gamma$ holds.
Hence, by the exact same line of arguments as in the proof of \cref{le:necessaryCond} we have the following lemma.
\begin{lemma}
\label{le:necessaryCond2}
Assume $\mathcal{E}_\OPT^q \wedge \mathcal{E}_\ALG^q$ holds.
Let $B$ be a batch in \OPT.
Let $I$ be such that $r_j \in I$ for all $j \in B$.
An overhead of at most $2s$ is associated to $B$ in \textsc{SetupEstimate}($I$).

Also, if the overhead associated to $B$ in \OPT is smaller than $2s$ and is $2s$ in the schedule of \textsc{SetupEstimate}($I$), $w(B) \geq \Gamma- F^* -s \eqqcolon \bar w$ needs to hold and jobs of $B$ with size at least $\bar w$ need to be released in an interval of length $\Gamma$.
\end{lemma}
In the next two lemmas we show that $N_s(I_i)$ is indeed a good estimation of the number of setups \OPT and \ALG have to perform, respectively.
\cref{le:saveSetups} essentially follows by \cref{le:necessaryCond2} together with the definition of $I_i$ to consist of $\mu$ many candidate intervals.
To prove \cref{le:correctEstimate} we exploit the fact that all jobs in $S_\Gamma^\delta$ have a flow time of at least $3\Gamma$ by \cref{le:workload} so that \ALG and \textsc{SetupEstimate} essentially behave in the same way.
\begin{lemma}
\label{le:saveSetups}
Assume $\mathcal{E}_\OPT^q \wedge \mathcal{E}_\ALG^q$ holds and consider $I_i$ for a fixed $i$.
For the overhead of \OPT it holds $\oho(I_i) \geq (N_s(I_i) -6\mu) s$.
\end{lemma}
\begin{proof}
Recall that by \cref{le:necessaryCond2} \OPT may have less overhead if it processes some jobs in one batch that are processed in two batches by \textsc{SetupEstimate}($I$).
However, a necessary condition for this is a workload of jobs of one type with size at least $\Gamma/4$ (in the unperturbed instance $\mathcal{I}$) and released in an interval of length $\Gamma$.
Let $\tilde C_1, \ldots, \tilde C_\mu$ be the candidates in $I_i$.
Associate all jobs released in $[l(\tilde C_1), l(\tilde C_1) + 2\Gamma]$ to candidate $\tilde C_1$ and inductively associate all jobs released in $[l(\tilde C_j), l(\tilde C_j) + 2\Gamma]$ not associated to a candidate $\tilde C_{j'}$ for $j' < j$ to $\tilde C_j$.
Note that by this construction, a workload of at most $3\Gamma$ (in the perturbed instance $\mathcal{\hat I}$) can be associated to each candidate $\tilde C_j$ as otherwise $F^* > \Gamma$ contradicting $\mathcal{E}_\OPT^q$.
Hence, taking the workload associated to $\tilde C_1, \ldots, \tilde C_\mu$, the necessary conditions of \cref{le:necessaryCond2} can be fulfilled at most $6\mu$ times and they cannot be fulfilled for any workload not associated to a candidate interval $\tilde C_j$.
Hence, we have $\oho(I_i) \geq (N_s(I_i) - 6\mu) s$, proving the lemma.
 \end{proof}

\begin{lemma}
\label{le:correctEstimate}
Consider an interval $I_i$ and suppose that all jobs from $I_i$ belong to $S^\delta_\Gamma$.
Then it holds $\oha(I_i) \leq N_s(I_i)s$, where $\oha(I_i)$ denotes the overhead of \ALG associated to jobs $j$ with $r_j \in I_i$.
\end{lemma}
\begin{proof}
Consider the subschedule $S'$ of \ALG starting with the first job from $I_i$ to which overhead is associated and ending with the last one to which overhead is associated.
Let $Z$ and $Z'$ be the sequences of jobs as induced by \textsc{SetupEstimate}($I_i$) and $S'$, respectively.
We remove all jobs not released during $I_i$ from $Z'$ and all jobs which are not part of $S'$ from $Z$.
Compare both resulting sequences $Z$ and $Z'$ and note that they consist of the exact same sets of jobs.
If both are identical, the lemma holds because no overhead can be associated to a job removed from $Z'$.

Hence, consider the case that $Z$ and $Z'$ differ and let $j$ and $j'$ be the jobs in $Z$ and $Z'$, respectively, at which both sequences differ the first time.
Then, in $Z$ job $j$ is preferred over job $j'$ and in $Z'$ job $j'$ is preferred over job $j$.
This can only be the case when $j'$ is scheduled by \ALG at a time $t$ such that $r_j > t$.
Because in $Z$ job $j$ is preferred over $j'$, it needs to hold $r_j \leq r_{j'} + \Gamma$.
But at the time $t$ at which $j'$ is scheduled it needs to hold $t \geq r_{j'}+3\Gamma - s - p_{j'}$ as otherwise its flow time is smaller than $3\Gamma$ which contradicts the assumption that it belongs to $S^\delta_\Gamma$ by \cref{le:workload}.
Hence, we obtain a contradiction as we have $r_j >t$ and $r_j \leq t$ and thus, $Z$ and $Z'$ are identical.
 \end{proof}

\subsection{Good and Bad Events}
\label{sec:events}
We are now ready to define good and bad events, which are outcomes of the perturbation of the job sizes that help the algorithm to achieve a small and help the adversary to achieve a high competitiveness, respectively.
Let $w_{\mathcal{I}}(I_i) \coloneqq \sum_{j : r_j \in I_i} p_j$ and $w_{\mathcal{\hat I}}(I_i) \coloneqq \sum_{j : r_j \in I_i} \hat p_j$ denote the workload released in the interval $I_i$ in instance $\mathcal{I}$ and $\mathcal{\hat I}$, respectively.
We distinguish two kinds of intervals $I_i$ and associate a good and a bad event to each of them.
We call an interval $I_i$ to be \emph{dense} if 
$w_{\mathcal{I}}(I_i) + N_s(I_i)s \geq |I_i| \label[ineq]{ineq:dense}$
and associate an event $\mathcal{D}^{\text{good}}_i$ or $\mathcal{D}^{\text{bad}}_i$ to $I_i$ depending on whether $w_{\mathcal{\hat I}}(I_i) \geq w_{\mathcal{I}}(I_i) +  \varepsilon^2 \Gamma/(18\sqrt{c_2} s \log n)$ holds or not.
Symmetrically, we call an interval $I_i$ to be \emph{sparse} if
$w_{\mathcal{I}}(I_i) + N_s(I_i)s <|I_i| \label[ineq]{ineq:sparse}$
and associate an event $\mathcal{S}_i^{\text{good}}$ or $\mathcal{S}_i^{\text{bad}}$ to $I_i$ depending on whether $w_{\mathcal{\hat I}}(I_i) \leq w_{\mathcal{I}}(I_i) - \varepsilon^2\Gamma /( 18\sqrt{c_2}s \log n)$ holds or not.

We next show two lemmas which upper bound $\Pr[\mathcal{E}_\OPT^q \wedge \mathcal{E}_\ALG^q]$ by the probability of occurrences of good events.
As we will see in \cref{le:finalLemma} this is sufficient as we can prove the respective good events to happen with sufficiently large probability.

\begin{lemma}
\label{le:highOpt}
$\Pr[\mathcal{E}_\OPT^q \wedge \mathcal{E}_\ALG^q]
\leq \Pr[\text{no event } \mathcal{D}^{\text{good}}_i \text{ happens}]$.
\end{lemma}

\begin{proof}
We show that if an event $\mathcal{D}^{\text{good}}_i$ happens, then $\mathcal{E}_\OPT^q$ does not hold.
Consider a dense interval $I_i$ and assume an event $\mathcal{D}^{\text{good}}_i$ occurs.
Then we have by definition of dense intervals and the definition of event $\mathcal{D}^{\text{good}}_i$ that $w_{\mathcal{I}}(I_i) + N_s(I_i)s \geq |I_i|$ and $w_{\mathcal{\hat I}}(I_i) \geq w_{\mathcal{I}}(I_i) +  \varepsilon^2 \Gamma/(18\sqrt{c_2} s \log n)$.
Taken together, $w_{\mathcal{\hat I}}(I_i) + N_s(I_i)s \geq |I_i| + \varepsilon^2 \Gamma/(18\sqrt{c_2} s \log n)$.
On the other hand, together with \cref{le:saveSetups} we then have $w_{\mathcal{\hat I}}(I_i)+\oho(I_i) \geq |I_i| +\varepsilon^2 \Gamma/(18\sqrt{c_2} s \log n) - 6s \cdot \mu$.
By \cref{obs:opt} we have
\begin{align*}
F^* & \geq w_{\mathcal{\hat I}}(I_i) + \oho(I_i) - |I_i| 
\geq \frac{\varepsilon^2\Gamma}{18\sqrt{c_2}s \log n} - 6s\left(\frac{2\varepsilon^2 \Gamma}{c_2 s^2 \log^2n}\right)\\
& \geq \frac{\varepsilon^2 \Gamma}{18\sqrt{c_2}s \log n} \left(1-\frac{12\cdot18}{\sqrt{c_2}\log n} \right) 
\geq \frac{c_1\alpha^{q-1} \log n}{18\sqrt{c_2}} \left(1-\frac{12\cdot 18}{\sqrt{c_2}\log n} \right)
> \alpha^{q+1}
\end{align*}
for sufficiently large $c_1 > c_2$ and $n$.
Then $\mathcal{E}_\OPT^q$ does not hold.
 \end{proof}
In \cref{le:finalLemma} we will see that the number $N_D$ of dense intervals in $\mathcal{I}$ can be bounded by $N_D = 7\log n$ as otherwise the probability for event $\mathcal{E}^q_\OPT$  to hold is only $1/n$.

Thus, next we consider the case $N_D < 7 \log n$.
Consider the sequence of events associated to sparse intervals.
A \emph{run} of events $\mathcal{S}_i^{\text{bad}}$ is a maximal subsequence such that no event $\mathcal{S}_i^{\text{good}}$ happens within this subsequence.

\begin{lemma}
\label{le:longRun}
If $N_D < 7 \log n$,
$
\Pr[\mathcal{E}_\OPT^q \wedge \mathcal{E}_\ALG^q]
\leq \Pr[\exists \text{ run of } \mathcal{S}_i^{\text{bad}} \text{ of length } \geq 14 \log n]$.
\end{lemma}

\begin{proof}
We assume that all runs of events $\mathcal{S}_i^{\text{bad}}$ are shorter than $14 \log n$ and $\mathcal{E}_\OPT^q \wedge \mathcal{E}_\ALG^q$ holds and show a contradiction.
From $\mathcal{E}_\ALG^q$ we can deduce by \cref{le:workload} that $S^\delta_{\Gamma}$ exists.
Since we will use the following reasoning iteratively, let $S = S^\delta_{\Gamma}$.
Using the terminology from \cref{le:workload}, let $j_1, j_2, \ldots, j_m$ be the jobs in $S$ and, as before, $\ell$ be the number of batches in $S$.
By \cref{le:workload} it needs to hold
$\sum_{i=1}^{\ell}w_\mathcal{\hat I}(B_i) + r_{j_1} - r_{j_m} \geq 3 \Gamma - (\ell-1)s$.
Let $I_{\iota+1}$ be the first interval $I_i$ such that $l(I_{\iota+1}) \geq r_{j_1}$.
Let $\kappa$ be chosen such that $\kappa$ is the smallest integer where in $I_{\iota+\kappa}$ an event $\mathcal{S}_{\iota+\kappa}^{\text{good}}$ occurs if $\kappa$ exists and otherwise set $\kappa$ such that $I_{\iota+\kappa}$ ends with $r_{max}$.
Note that it holds $\kappa < 21 \log n$ because of the assumption $N_D < 7 \log n$ and the length of the longest run.
Let $I_{\iota} = [\min_{1 \leq i \leq m} r_{j_i}, l(I_{\iota+1}))$.
We claim that all jobs belonging to $\bigcup_{i=0}^\kappa I_{\iota+i}$ need to have a flow time below $\alpha \Gamma$.
Assume this is not the case.
We have a contradiction as 
\begin{align*}
F^* &\geq w_\mathcal{\hat I}(I(S)) + \oho(I(S)) - |I(S)| \\
& \geq \sum_{i=1}^{\ell}w_\mathcal{\hat I}(B_i) + \oho(I(S)) +r_{j_1} - r_{j_m} - 2\Gamma\\
& \geq \sum_{i=1}^{\ell}w_\mathcal{\hat I}(B_i) +r_{j_1} - r_{j_m} - 2\Gamma + \oha(I(S)) -\frac{21 \log n \cdot 12s \varepsilon^2 \Gamma}{c_2 s^2\log^2n} \\
& \geq \Gamma -\frac{252 \log n s \varepsilon^2 \Gamma}{c_2 s^2\log^2n}
\geq \varepsilon^2\Gamma \left(\frac{1}{\varepsilon^2}-\frac{252}{c_2s\log n}\right)
 \geq \frac{1}{2}\varepsilon^2c_1\varepsilon^{-2}\alpha^{q-1}s \log^2 n
 > \alpha^{q+1}
\end{align*}
where we used \cref{obs:opt} in the first inequality, the fact that $|I(S)| \leq (r_{j_m}-r_{j_1})+2\Gamma$ in the second, \cref{le:saveSetups,le:correctEstimate} in the third, \cref{le:workload} in the fourth and in the remaining inequalities suitable values for $c_1 > c_2$ and the fact that $\Gamma \geq c_1\varepsilon^{-2}\alpha^{q-1}s\log^2 n$.
Observe that in case $r(I_{\iota+\kappa}) = r_{max}$, we are done as $\mathcal{E}_\ALG^q$ cannot hold.

Otherwise, 
consider the situation directly before the first job $\tilde j$ with $r_{\tilde j} > r(I_{\iota+\kappa})$ is started.
Denote the subschedule of $S$ up to (not including) job $\tilde j$ by $\tilde S$.
Let $\oha(I)$ be the overhead in $S$ associated to jobs released in the interval $I$.
Let $\oha(I, \tilde S)$ and $\oha(I, \neg \tilde S)$ be the overhead of jobs released in interval $I$ and which are part and not part of $\tilde S$, respectively.
Let $w_\mathcal{\hat I}(I, \tilde S)$ and $w_\mathcal{\hat I}(I, \neg \tilde S)$  be the workload of jobs released in interval $I$ and which are part and not part of $\tilde S$, respectively.
For brevity let $L =  w_\mathcal{\hat I}([0,r_{j_1}), \tilde S) - w_\mathcal{\hat I}([r_{j_1},r_{\tilde j}), \neg \tilde S) + \oha([0,r_{j_1}), \tilde S)  - \oha([r_{j_1},r_{\tilde j}), \neg \tilde S)$.
We can then bound the workload and setups in $\tilde S$ by
\begin{align*}
& w_\mathcal{\hat I}(\tilde S) + \oha(\tilde S)\\
 \leq & w_\mathcal{\hat I}([r_{j_1}, l(I_{\iota+\kappa}))) + w_\mathcal{\hat I}(I_{\iota+\kappa}) + w_\mathcal{\hat I}([0,r_{j_1}), \tilde S) - w_\mathcal{\hat I}([r_{j_1},r_{\tilde j}), \neg \tilde S)\\ 
 &+ \oha([r_{j_1}, l(I_{\iota+\kappa})))  + \oha(I_{\iota+\kappa})\\
 & + \oha([0,r_{j_1}), \tilde S)  - \oha([r_{j_1},r_{\tilde j}), \neg \tilde S)\\
 \leq & l(I_{\iota+\kappa}) - r_{j_1}  + F^* + 21 \log n 12 s \cdot \frac{\varepsilon^2 \Gamma}{c_2s^2 \log^2 n} + |I_{\iota+\kappa}| - \frac{\varepsilon^2\Gamma}{18\sqrt{c_2}s \log n} +L \\
= & r(I_{\iota+\kappa}) -r_{j_1} + F^* + \frac{\varepsilon^2\Gamma}{s \log n 18\sqrt{c_2}}\left(\frac{252 \cdot 18}{\sqrt{c_2}}-1\right) + L\\
\leq & r(I_{\iota+\kappa}) -r_{j_1} + F^* + \frac{c_1\alpha^{q-1} \log n}{18\sqrt{c_2}}\left(\frac{252 \cdot 18}{\sqrt{c_2}}-1\right) + L
<r(I_{\iota+\kappa}) -r_{j_1} - 2 F^* + L,
\end{align*}
where we used \cref{obs:opt} together with \cref{le:saveSetups} and the fact that to $I_{\iota+\kappa}$ an event $\mathcal{S}^\text{good}_{\iota+\kappa}$ is associated in the second inequality, the lower bound on $\Gamma$ in the third inequality and suitable values for $c_1$ and $c_2$ in the last inequality.
Then, job $\tilde j$ is started before $r_{j_1} + F_{j_1} + r(I_{\iota+\kappa}) -r_{j_1} -2F^* +L + s$ and finished by $r(I_{\iota+\kappa}) + F_{j_1} +L$ with flow time $F_{\tilde j} \leq F_{j_1} + L$.
For $S= S_{\Gamma}^\delta$ we have $L \leq - w_\mathcal{\hat I}([r_{j_1},r_{\tilde j}), \neg \tilde S) - \oha([r_{j_1},r_{\tilde j}), \neg \tilde S)$ as no jobs with smaller release time than $r_{j_1}$ can be part of $S$.
Thus,  $F_{\tilde j} < (\alpha-\delta)\Gamma- w_\mathcal{\hat I}([r_{j_1},r_{\tilde j}), \neg \tilde S) - \oha([r_{j_1},r_{\tilde j}), \neg \tilde S)$.
Now, applying the same arguments with $S = S_{\Gamma}^\delta(\tilde j)$ and using \cref{cor:workload} instead of \cref{le:workload}, we find a further job with flow time at most $(\alpha-\delta)\Gamma$ (and all jobs processed before have flow time below $\alpha \Gamma$).
Iterating this process we will eventually reach the end of the instance without finding a job with flow time at least $\alpha \Gamma$, contradicting that $\mathcal{E}_\ALG^q$ holds.
Formally, it remains to prove that the claim $F_{\tilde j} \leq (\alpha-\delta)\alpha^q$ also holds for later iterations.
 We introduce the following notations.
 Denote by $j_1^0$ and $\tilde{j}^0$ the jobs $j_1$ and $\tilde{j}$ from the first iteration as in the main body of the paper, respectively.
 For the following iterations, we use the notation $j_1^i$ and $\tilde{j}^i$ for the respective jobs of the $i$-th iteration.
 Note that $j_1^{i}=\tilde{j}^{i-1}$ and we will thus only use $j_1^0$, but $\tilde{j}^i$ at all other places.
 Similarly, we denote $\tilde{S}^i$ as the symbol $\tilde{S}$ from the $i$-th iteration.
 We define $w_{\hat{\mathcal{I}}}(I,\bigwedge_{i=0}^\nu \neg \tilde{S})$ as the natural extension of the prior definition to be the workload of jobs released in interval $I$ and which are not part of any of the subschedules $\tilde{S}^0,\ldots,\tilde{S}^\nu$.
 For $\oha(I,\bigwedge_{i=0}^\nu \neg \tilde{S})$, the extension is defined similarly.
 We now prove the following claim inductively:
 \[F_{\tilde{j}^\nu}\leq (\alpha-\delta)\Gamma-w_{\hat{\mathcal{I}}}\left([r_{j_1^0},r_{\tilde{j}^\nu}),\bigwedge_{i=0}^{\nu}\neg\tilde{S}^i\right)-\oha\left([r_{j_1^0},r_{\tilde{j}^\nu}),\bigwedge_{i=0}^{\nu}\neg\tilde{S}^i\right).\]
 For ease of notation, we introduce the combined expression of $\comb(I,X)\coloneqq w_{\hat{\mathcal{I}}}(I,X) + \oha(I,X)$.
 
 As we have already seen the induction base, assume the claim is true for $\nu-1$.
 By using \[\comb\left([r_{j_1^0},r_{\tilde{j}^{\nu-1}}),\bigwedge_{i=0}^{\nu-1}\neg\tilde{S}^i\right) = \comb\left([r_{j_1^0},r_{\tilde{j}^{\nu-1}}),\bigwedge_{i=0}^{\nu}\neg\tilde{S}^i\right) + \comb\left([r_{j_1^0},r_{\tilde{j}^{\nu-1}}),\bigwedge_{i=0}^{\nu-1}\neg\tilde{S}^i\wedge \tilde{S}^i\right)\]
 and $\comb\left([r_{j_1^0},r_{\tilde{j}^{\nu-1}}),\bigwedge_{i=0}^{\nu}\neg\tilde{S}^i\wedge \tilde{S}^i\right) = \comb\left([r_{j_1^0},r_{\tilde{j}^{\nu-1}}),\tilde{S}^i\right)$,
 we estimate
 \begin{align*}
F_{\tilde{j}^\nu}&\leq F_{\tilde{j}^{\nu-1}} + L\\
 &\leq (\alpha-\delta)\Gamma - \comb\left([r_{j_1^0},r_{\tilde{j}^{\nu-1}}),\bigwedge_{i=0}^{\nu-1}\neg\tilde{S}^i\right) \\ 
 &\qquad + \comb([0,r_{\tilde{j}^{\nu-1}}),\tilde{S}^\nu)-\comb([r_{\tilde{j}^{\nu-1}},r_{\tilde{j}^{\nu}}),\neg\tilde{S}^\nu) \\
 &= (\alpha-\delta)\Gamma - \comb\left([r_{j_1^0},r_{\tilde{j}^{\nu-1}}),\bigwedge_{i=0}^{\nu}\neg\tilde{S}^i\right) - \comb\left([r_{j_1^0},r_{\tilde{j}^{\nu-1}}),\tilde{S}^\nu\right)\\ 
 &\qquad + \underbrace{\comb([0,r_{j_1^0}),\tilde{S}^\nu)}_{=0} + \comb([r_{j_1^0},r_{\tilde{j}^{\nu-1}}),\tilde{S}^\nu) - \comb([r_{\tilde{j}^{\nu-1}},r_{\tilde{j}^{\nu}}),\neg\tilde{S}^\nu)\\
 &\leq (\alpha-\delta)\Gamma-\comb\left([r_{j_1^0},r_{\tilde{j}^{\nu-1}}),\bigwedge_{i=0}^{\nu}\neg\tilde{S}^i\right) - \comb([r_{\tilde{j}^{\nu-1}},r_{\tilde{j}^{\nu}}),\bigwedge_{i=0}^{\nu}\neg\tilde{S}^i) \\
 &=(\alpha-\delta)\Gamma-\comb\left([r_{j_1^0},r_{\tilde{j}^\nu}),\bigwedge_{i=0}^{\nu}\neg\tilde{S}^i\right) \enspace .
 \end{align*}
 The claim follows.
 Also, as $\comb\left([r_{j_1^0},r_{\tilde{j}^\nu}),\bigwedge_{i=0}^{\nu}\neg\tilde{S}^i\right)$ is always non-negative, the claim implies $F_{\tilde{j}^\nu}\leq (\alpha-\delta)\Gamma$.
 \end{proof}
To finally bound the probability of good events to happen, we need the following lemma.

\begin{lemma}
\label{le:probabilites}
Let $J$ be a set of jobs and assume that processing times are perturbed according to a uniform or normal smoothing distribution.
With probability at least $1/10$, $w_\mathcal{\hat I}(J) \geq w_\mathcal{I}(J) + \frac{\varepsilon}{5} (\lfloor w_\mathcal{I}(J) \rfloor/3)^{0.5}$.
Also, with probability at least $1/10$, $w_\mathcal{\hat I}(J) \leq w_\mathcal{I}(J) - \frac{\varepsilon}{5} (\lfloor w_\mathcal{I}(J)\rfloor /3)^{0.5}$.
\end{lemma}
\begin{proof}
We first show the lemma for the case of a uniform smoothing distribution and then continue with the case of a normal distribution.

\emph{Uniform Smoothing Distribution.  }
We can describe the perturbed workload as $w_\mathcal{\hat I}(J) = w_\mathcal{I}(J) + X$ where $X$ is the random variable given by $X = \sum_{j \in J} X_j$ where $X_j \sim \mathcal{U}(-\varepsilon p_j, \varepsilon p_j)$.
Let $w = \lfloor w_\mathcal{I}(J) \rfloor$.
We distinguish two cases depending on whether there exists a job $j' \in J$ with $p_{j'} \geq \frac{2}{5} \varepsilon \frac{\sqrt{w}}{\sqrt{3}}$.
In the positive case, we have $\Pr[X_{j'} \geq \frac{1}{5}\varepsilon (w/3)^{0.5}] \geq 1/4$ and $\Pr[\sum_{j \in J \setminus \{j'\}} X_j \geq 0] \geq 1/2$.
Hence, in this case the lemma holds.
Consider the case in which for all $j \in J$ it holds $p_{j} <\frac{2}{5} \varepsilon \frac{\sqrt{w}}{\sqrt{3}}$.
We then have $\mathbb{E}(X_j) = 0$ and $\mathbb{V}(X_j) = \sigma^2_j = \frac{1}{3}(\varepsilon p_j)^2 \geq \frac{1}{3}\varepsilon^2$, for all $j$.
Also $\mathbb{E}[|X_j|^3] = \frac{1}{4} (\varepsilon p_j)^3$.
We now bound the probability we are interested in by a normal approximation.
Let $S = \frac{X_1 + \ldots + X_{|J|}}{\sqrt{\sigma_1^2 + \ldots + \sigma_{|J|}^2}}$, $F$ be the distribution of $S$ and $\delta = \sup_x |F(x) - \Phi_{0,1}(x)|$, where $\Phi_{0,1}(x)$ is the distribution function of the standard normal distribution.
By the central limit theorem we have 
\begin{align*}
& \Pr\left[X_1 + \ldots + X_{|J|} \leq  \frac{\varepsilon}{5} \left(\lfloor w_\mathcal{I}(J) \rfloor/3\right)^{0.5}\right]
\leq  \Pr\left[X_1 + \ldots + X_{|J|} \leq  \frac{1}{5} \left(\sum \sigma_i^2\right)^{0.5}\right] \\
\leq &\Phi_{0,1}\left(\frac{1}{5}\right) + \delta \leq  0.57926 + \delta \enspace .
\end{align*}
Also, we can bound $\delta$ using standard Berry-Esseen bounds by
\begin{align*}
\delta \leq 0.56\left(\sum \sigma_i^2\right)^{-\frac{1}{2}} \cdot \max \frac{\mathbb{E}[|X_i|^3]}{\sigma_i^2}
\leq 1/(\varepsilon \sqrt{w/3}) \cdot \frac{3}{4} \varepsilon \frac{2}{5} \sqrt{w/3}
< 3/10 \enspace .
\end{align*}
Together with the symmetry of the uniform distribution, we obtain the lemma for uniform perturbations.

\emph{Normal Smoothing Distribution.  }
Recall that $\varepsilon^2 = 2.64\sigma^2$.
We can describe the perturbed workload as $w_{\mathcal{\hat I}}(J) = w_\mathcal{I}(J) + X$ where $X$ is a random variable given by $X = \sum_{j \in J}X_j$ and $X_j \sim \mathcal{N}_{(-p_j,p_j)}(0, (\sigma p_j)^2)$.
We have $\mathbb{E}[X_j] = 0$ and $\mathbb{V}[X_j] = (\sigma p_j)^2(1-\frac{2/\sigma\phi(2/\sigma)}{2\Phi(2/\sigma)-1}) \geq (\sigma p_j)^2(1- \frac{0.11}{0.95}) \geq 0.88(\sigma p_j)^2$.
Also we have 
\[
\mathbb{E}[|X_j|^3] = \frac{1}{\sqrt{2\pi}\sigma(\Phi(\frac{p_j}{\sigma})-\Phi(\frac{-p_j}{\sigma}))}\int_{-p_j}^{p_j} \! |x|^3 \exp(-0.5(\frac{x}{\sigma})^2) \, \mathrm{d}x \leq \frac{1}{\sqrt{2\pi}\sigma 0.68} 4\sigma^4 \leq 2.35\sigma^3 \enspace .\]

By the central limit theorem we have 
\begin{align*}
& \Pr\left[X_1 + \ldots + X_{|J|} \leq  \frac{\varepsilon}{5} \left(\lfloor w_\mathcal{I}(J) \rfloor/3\right)^{0.5}\right]
= \Pr\left[X_1 + \ldots + X_{|J|} \leq  \frac{1}{5} \left(0.88\sigma^2\lfloor w_\mathcal{I}(J) \rfloor\right)^{0.5}\right]\\
\leq &\Pr\left[X_1 + \ldots + X_{|J|} \leq  \frac{1}{5} \left(\sum \mathbb{V}[X_i]\right)^{0.5}\right]
\leq \Phi_{0,1}\left(\frac{1}{5}\right) + \delta \leq  0.57926 + \delta \enspace .
\end{align*}
Also, we can bound $\delta$ using standard Berry-Esseen bounds by
\begin{align*}
\delta \leq 0.56\left(\sum \mathbb{V}[X_i]\right)^{-\frac{1}{2}} \cdot \max \frac{\mathbb{E}[|X_i|^3]}{\sigma_i^2}
\leq 0.597 \frac{1}{\sqrt{|J|}\sigma} \cdot 2.6705\sigma \leq 1.5942885/\sqrt{|J|} < 3/10 \enspace ,
\end{align*}
if $|J| \geq 29$.
Note that we can assume $|J| \geq 29$ as we only apply the bound in \cref{le:finalLemma} and thus under the assumption that $w_\mathcal{I}(J) \geq \mu \frac{\Gamma}{4} \geq \mu \frac{c_1\varepsilon^{-2}\log^2 n s}{4\alpha^2} F^* = \Omega(\log^2n)F^*$, which requires at least $\Omega(\log^2n)$ many jobs. 
Together with the symmetry of the normal distribution, we obtain the lemma.
 \end{proof}

\begin{theorem}
\label{le:finalLemma}
The smoothed competitiveness of \ALG is $O(\sigma^{-2} \log^2n)$ when processing times $p_j$ are perturbed independently at random to $\hat p_j = (1 + X_j)p_j$ where $X_j \sim \mathcal{U}(-\varepsilon, \varepsilon)$ or $X_j \sim \mathcal{N}_{(-1,1)}(0,\sigma^2)$.
\end{theorem}

\begin{proof}
Recall that it only remains to prove
$\Pr[\mathcal{E}_\OPT^q \wedge \mathcal{E}_\ALG^q] \leq 1/n$.
First consider the case $N_D \geq 7 \log n$.
For a fixed $i$ we have $\Pr[\mathcal{D}_i^\text{good}] \geq \frac{1}{10}$ because of the following reasoning.
According to \cref{le:probabilites}, it holds $\Pr[w_\mathcal{\hat I}(I_i) \geq w_\mathcal{I}(I_i) + \frac{\varepsilon}{5} (\lfloor w_\mathcal{I}(I_i)\rfloor/3)^{0.5}] \geq \frac{1}{10}$.
By definition of $I_i$ we can bound $\frac{\varepsilon}{5} (\lfloor w_\mathcal{I}(I_i)\rfloor/3)^{0.5} \geq \frac{\varepsilon}{5} (\frac{\mu\Gamma}{12})^{0.5} \geq \varepsilon^2 \frac{\Gamma}{18\sqrt{c_2}s \log n}$ which implies 
$\Pr[\mathcal{D}_i^\text{good}] \geq \frac{1}{10}$.
Because $N_D \geq 7 \log n$, the probability that no event $\mathcal{D}_i^\text{good}$ occurs is then upper bounded by $(1-\frac{1}{10})^{7\log n} \leq 1/n$ and so is $\Pr[\mathcal{E}_\OPT^q \wedge \mathcal{E}_\ALG^q]$ according to \cref{le:highOpt}.

For the case $N_D < 7 \log n$ the same line of arguments gives $\Pr[\mathcal{S}_i^\text{good}]\geq 1/10$ for each sparse interval $I_i$.
Hence, the probability for a run of events $\mathcal{S}_i^{\text{bad}}$ of length at least $14\log n$ is at most $1/n$ and so is $\Pr[\mathcal{E}_\OPT^q \wedge \mathcal{E}_\ALG^q]$ by \cref{le:longRun}.
 \end{proof}

To conclude, \cref{le:finalLemma} shows a polylogarithmic smoothed competitiveness, which significantly improves upon the worst-case bound of $\Theta(\sqrt{n})$.
It would be very interesting for future work to further investigate if it is possible to improve the smoothed analysis of \ALG.
Although we were not able to show such a result, it is quite possible that the actual smoothed competitiveness is independent of $n$.
Though, proving such a result would probably require a different approach; the $\log n$ term in our result seems to be inherent to our analysis as it relies on the length of the longest run of bad events, each occurring with constant probability.

\bibliographystyle{plain}
\bibliography{references}

\begin{thebibliography}{10}

\bibitem{ali3}
Ali Allahverdi.
\newblock The third comprehensive survey on scheduling problems with setup
  times/costs.
\newblock {\em European Journal of Operational Research}, 246(2):345--378,
  2015.

\bibitem{ali1}
Ali Allahverdi, Jatinder~ND Gupta, and Tariq Aldowaisan.
\newblock A review of scheduling research involving setup considerations.
\newblock {\em Omega}, 27(2):219--239, 1999.

\bibitem{ali2}
Ali Allahverdi, C.~T. Ng, T.~C.~Edwin Cheng, and Mikhail~Y. Kovalyov.
\newblock A survey of scheduling problems with setup times or costs.
\newblock {\em European Journal of Operational Research}, 187(3):985--1032,
  2008.

\bibitem{journalMaxFlowtime}
S.~Anand, Karl Bringmann, Tobias Friedrich, Naveen Garg, and Amit Kumar.
\newblock {Minimizing Maximum (Weighted) Flow-Time on Related and Unrelated
  Machines}.
\newblock {\em Algorithmica}, 77(2):515--536, 2017.

\bibitem{journalMaxFlowtime2}
Nikhil Bansal and Bouke Cloostermans.
\newblock {Minimizing Maximum Flow-Time on Related Machines}.
\newblock {\em Theory of Computing}, 12(1):1--14, 2016.

\bibitem{smoothedComp}
Luca Becchetti, Stefano Leonardi, Alberto Marchetti{-}Spaccamela, Guido
  Sch{\"{a}}fer, and Tjark Vredeveld.
\newblock {Average-Case and Smoothed Competitive Analysis of the Multilevel
  Feedback Algorithm}.
\newblock {\em Mathematics of Operations Research}, 31(1):85--108, 2006.

\bibitem{flowtime}
Michael~A. Bender, Soumen Chakrabarti, and S.~Muthukrishnan.
\newblock {Flow and Stretch Metrics for Scheduling Continuous Job Streams}.
\newblock In {\em Proceedings of the 9th Annual {ACM-SIAM} Symposium on
  Discrete Algorithms (SODA '98)}, pages 270--279. ACM/SIAM, 1998.

\bibitem{priority}
Allan Borodin, Morten~N. Nielsen, and Charles Rackoff.
\newblock {(Incremental) Priority Algorithms}.
\newblock {\em Algorithmica}, 37(4):295--326, 2003.

\bibitem{saksReport}
Srikrishna Divakaran and Michael Saks.
\newblock An online scheduling problem with job set-ups.
\newblock Technical report, DIMACS Technical Report, 2000.

\bibitem{saks}
Srikrishnan Divakaran and Michael~E. Saks.
\newblock {An Online Algorithm for a Problem in Scheduling with Set-ups and
  Release Times}.
\newblock {\em Algorithmica}, 60(2):301--315, 2011.

\bibitem{compAlter1}
Benjamin Hiller and Tjark Vredeveld.
\newblock Probabilistic alternatives for competitive analysis.
\newblock {\em Computer Science - R{\&}D}, 27(3):189--196, 2012.

\bibitem{traffic}
Micha Hofri and Keith~W. Ross.
\newblock On the optimal control of two queues with server setup times and its
  analysis.
\newblock {\em {SIAM} Journal on Computing}, 16(2):399--420, 1987.

\bibitem{land}
Klaus Jansen and Felix Land.
\newblock {Non-preemptive Scheduling with Setup Times: {A} {PTAS}}.
\newblock In {\em Proceedings of the 22nd International Conference on Parallel
  and Distributed Computing (Euro-Par '16)}, pages 159--170. Springer, 2016.

\bibitem{compAlter3}
Elias Koutsoupias and Christos~H. Papadimitriou.
\newblock Beyond competitive analysis.
\newblock {\em {SIAM} Journal on Computing}, 30(1):300--317, 2000.

\bibitem{compAlter2}
Alejandro L{\'{o}}pez{-}Ortiz.
\newblock Alternative performance measures in online algorithms.
\newblock In {\em Encyclopedia of Algorithms}, pages 67--72. Springer, 2016.

\bibitem{wads}
Alexander M{\"{a}}cker, Manuel Malatyali, Friedhelm {Meyer auf der Heide}, and
  S{\"{o}}ren Riechers.
\newblock {Non-preemptive Scheduling on Machines with Setup Times}.
\newblock In {\em Proceedings of the 14th International Symposium on Algorithms
  and Data Structures (WADS '15)}, pages 542--553. Springer, 2015.

\bibitem{fifo}
Monaldo Mastrolilli.
\newblock {Scheduling to Minimize Max Flow Time: Offline and Online
  Algorithms}.
\newblock In {\em Proceedings of the 14th International Symposium on
  Fundamentals of Computation Theory (FCT '03)}, pages 49--60. Springer, 2003.

\bibitem{monma}
Clyde~L. Monma and Chris~N. Potts.
\newblock {On the Complexity of Scheduling with Batch Setup Times}.
\newblock {\em Operations Research}, 37(5):798--804, 1989.

\bibitem{nonclairvoyant}
Rajeev Motwani, Steven Phillips, and Eric Torng.
\newblock Non-clairvoyant scheduling.
\newblock {\em Theoretical Computer Science}, 130(1):17--47, 1994.

\bibitem{manufacturing}
Chris~N. Potts.
\newblock Scheduling two job classes on a single machine.
\newblock {\em Computers {\&} {OR}}, 18(5):411--415, 1991.

\bibitem{operator}
Vinod~K. Sahney.
\newblock Single-server, two-machine sequencing with switching time.
\newblock {\em Operations Research}, 20(1):24--36, 1972.

\bibitem{taskSystems}
Guido Sch{\"{a}}fer and Naveen Sivadasan.
\newblock Topology matters: Smoothed competitiveness of metrical task systems.
\newblock {\em Theoretical Computer Science}, 341(1-3):216--246, 2005.

\bibitem{averageComp}
Mark Scharbrodt, Thomas Schickinger, and Angelika Steger.
\newblock {A New Average Case Analysis for Completion Time Scheduling}.
\newblock {\em Journal of the {ACM}}, 53(1):121--146, 2006.

\end{thebibliography}

\end{document}